\documentclass[journal]{IEEEtran}
\usepackage[T1]{fontenc}

\def\argmax{\mathop{\rm \arg\!\max}}

\usepackage{graphicx}
\usepackage[noadjust]{cite}
\usepackage{mcite}
\usepackage{amsfonts,helvet}
\usepackage{fancyhdr}
\usepackage{threeparttable}
\usepackage{epsf,epsfig}
\usepackage{amsthm}
\usepackage{amsmath}
\usepackage{siunitx}
\usepackage{amssymb}
\usepackage{dsfont}
\usepackage{subfigure}
\usepackage{color}
\usepackage{enumerate}
\usepackage{hyperref}
\usepackage{cancel}
\usepackage{bbm}
\usepackage{dsfont}
\usepackage[subnum]{cases}
\usepackage{adjustbox}
\usepackage[linesnumbered,ruled,noend]{algorithm2e}

\newtheorem{theorem}{Theorem}

\newtheorem{corollary}{Corollary}

\newtheorem{lemma}{Lemma}

\newmuskip\pFqmuskip

\def\bh{{\bf h}}

\def\bn{{\bf n}}

\def\bq{{\bf q}}

\def\bw{{\bf w}}

\def\by{{\bf y}}

\def\bI{{\bf I}}

\def\bR{{\bf R}}

\def\cA{\mbox{$\mathcal{A}$}}

\def\cC{\mbox{$\mathcal{C}$}}

\def\cN{\mbox{$\mathcal{N}$}}

\def\cR{\mbox{$\mathcal{R}$}}

\def\cT{\mbox{$\mathcal{T}$}}

\def\bbE{\mbox{$\mathbb{E}$}}

\def\bbV{\mbox{$\mathbb{V}$}}

\usepackage{array}

\makeatletter
\newcommand{\thickhline}{%
    \noalign {\ifnum 0=`}\fi \hrule height 1pt
    \futurelet \reserved@a \@xhline
}
\newcolumntype{"}{@{\hskip\tabcolsep\vrule width 1pt\hskip\tabcolsep}}
\makeatother

\IEEEoverridecommandlockouts

\title{Analysis of Ergodic Rate for Transmit Antenna Selection in Low-Resolution ADC Systems}

\author{
Jinseok Choi and Brian L. Evans \thanks{
Copyright (c) 2015 IEEE. Personal use of this material is permitted. However, permission to use this material for any other purposes must be obtained from the IEEE by sending a request to pubs-permissions@ieee.org.

J. Choi, and B. L. Evans are with the Wireless Networking and Communication Group, Department of Electrical and Computer Engineering, The University
of Texas at Austin, Austin, TX 78701 USA. (e-mail: jinseokchoi89@utexas.edu, bevans@ece.utexas.edu). 
The authors were supported by gift funding from Huawei Technologies.
}
}

\begin{document}
\maketitle
\thispagestyle{empty}

\begin{abstract}
In this paper, we analyze the ergodic rate of single transmit antenna selection (TAS) in low-resolution analog-to-digital converter (ADC) systems.
Using low-resolution ADCs is a potential power-reduction solution for multiple antenna systems.
Low-resolution ADC systems with TAS can further reduce cost and power consumption in wireless transceivers.
Considering such systems, we derive the approximated lower bound of the ergodic rate with TAS.
Here, we exploit the approximated distribution of the sum of Weibull random variables to address the challenge involved in analyzing the quantization error.
We, then, derive the approximated ergodic rate with TAS for a single receive antenna in closed form, which reveals the TAS gain in low-resolution ADC systems.
The upper bound of a single transmit and single receive antenna system under coarse quantization is derived to compare with the ergodic rate of the TAS system.
The analysis shows that the TAS method achieves a large improvement in ergodic rate with a moderate number of transmit antennas.
Simulation results validate the derived ergodic rates and resulting intuition.
\end{abstract}



\vspace{-1 em}
\section{Introduction}
\label{sec:intro}





Employing low-resolution ADCs has been widely studied to reduce the power consumption of base stations and mobile devices \cite{mezghani2007ultra,mo2015capacity}.
With the advent of large-scale antenna systems such as massive multiple-input multiple-output (MIMO) for sub-6GHz \cite{marzetta2010noncooperative, larsson2014massive} and millimeter wave (mmWave) communications \cite{pi2011introduction, rappaport2013millimeter}, the importance of low-resolution ADCs has greatly increased.
Consequently, low-resolution ADC systems have been considered as a possible low-power solution for future wireless systems \cite{jacobsson2017throughput} with their large number of antenna outputs coupled with wide bandwidth.

For low-resolution ADC systems, the analysis of achievable rate was performed in \cite{orhan2015low, fan2015uplink, jacobsson2017throughput}.
In \cite{orhan2015low}, the achievable rate for the mmWave systems was investigated to show the tradeoff between the resolution of ADCs and operating bandwidth.
In \cite{fan2015uplink,jacobsson2017throughput}, the approximated ergodic rates of uplink massive MIMO systems with linear receivers were derived in closed form. 
User scheduling in uplink low-resolution ADC systems was also investigated in \cite{choi2018user}.
Recently, a receive antenna selection algorithm was generalized to different quantization resolutions in \cite{choi2018antenna} to incorporate the quantization error penalty.
Transmit antenna selection (TAS), however, is still understudied in low-resolution ADC systems.
Although the benefits of TAS were widely analyzed in infinite-resolution ADC systems \cite{sanayei2007capacity,lin2012performance}, 
the analysis cannot capture the impact of coarse quantization. 
{\color{black} Consequently, considering the coarse quantization effect in the TAS analysis is still an open question.}

In this paper, we investigate TAS with a maximum ratio combining (MRC) receiver in low-resolution ADC systems and derive the approximated lower bound of the ergodic rate for single TAS.
A single transmit antenna which maximizes the achievable rate is selected.
Single TAS problems were investigated for infinite-resolution ADC systems in \cite{chen2005analysis,zhang2007outage,zhang2008performance}. 
Unlike the infinite-resolution ADC case, the received signal-to-noise ratio (SNR) with low-resolution ADCs involves a quantization error, which is a function of selected channels that makes the analysis particularly challenging.
To address such a challenge, we adopt an approximated probability density function (PDF) and cumulative distribution function (CDF) of the sum of Weibull distributions, and derive the lower bound of ergodic rate with TAS. 
To further analyze the TAS gain, we derive the approximated ergodic rate of TAS for a single receive antenna. 
The derived ergodic rates reveal that under coarse quantization, the TAS gain for a single receive antenna is limited because selecting an antenna with a larger channel gain also increases the quantization error.
The analysis, however, still shows that the TAS method achieves a large improvement in ergodic rate with a moderate number of transmit antennas.
The simulation results validate the derived ergodic rate expressions and resulting intuition.

{\color{black} The TAS can be used for both the downlink and uplink communications where a transmitter with multiple antennas and a single radio frequency (RF) chain selects a single antenna and a receiver with multiple antennas and multiple RF chains quantizes received signals with low-resolution ADCs to reduce cost and power consumption of the transmitter and receiver.
In this paper, we focus on a single user communication with the TAS method and an extension to a multiuser communication is a potential direction of future work.}

{\color{black}
{\it Notation}: $\bf{A}$ is a matrix and $\bf{a}$ is a column vector. 
$\mathbf{A}^{H}$ and $\mathbf{A}^T$  denote conjugate transpose and transpose, and $ \mathbf{a}_i$ indicates the $i$th column vector of $\bf A$. 
We denote $a_{i,j}$ as the $\{i,j\}$th element of $\bf A$ and $a_{i}$ as the $i$th element of $\bf a$. 
$\mathcal{CN}(\mu, \sigma^2)$ is the complex Gaussian distribution with mean $\mu$ and variance $\sigma^2$. 
$\mathbb{E}[\cdot]$ and $\bbV[\cdot]$ represent an expectation and variance operators, respectively.
The correlation matrix is denoted as ${\bf R}_{\bf xy} = \mathbb{E}[{\bf x}{\bf y}^H]$.
The diagonal matrix $\rm diag\{\bf A\}$ has $\{a_{i,i}\}$ at its $i$th diagonal entry.
${\bf I}_N$ denotes the $N$ dimensional identity matrix.
$\bf 0$ denotes a zero vector with a proper dimension.
$\|\bf A\|$ represents $L_2$ norm. 
$|\cdot|$ indicates an absolute value and cardinality for a scalar value $a$ and a set $\cA$, respectively.
}
\section{System Model}
\label{sec:sys_model}

We consider a wireless link in which a transmitter equipped with $N_t$ antennas selects a single transmit antenna.
A receiver equipped with $N_r$ antennas employs low-resolution ADCs with $b$ quantization bits.
With the transmit power $p_t$, the received baseband analog signal is given as ${\bf y} \!=\!\sqrt{p_t}{\bf h}_{i}{s}\! +\! {\bf n}$,
where ${s}$ and ${\bf n}$ denote the zero mean and unit variance transmit symbol and the additive white Gaussian noise (AWGN) vector ${\bf n} \sim \cC\cN({\bf 0}, \sigma^2\bI_{N_r})$, respectively, and ${\bf h}_{i}$ is the channel vector that corresponds to the transmit antenna $i$.
We consider a Rayleigh fading channel $\bh_{i} \sim \cC\cN({\bf 0},\bI_{N_r})$.

Each real and imaginary component of the complex output $y_i$ is quantized by the ADCs.
Since an additive quantization noise model (AQNM) with scalar gain $\alpha$ \cite{fletcher2007robust} shows reasonable accuracy in low and medium SNR ranges \cite{orhan2015low}, we adopt the AQNM to linearize the quantization process as a function of quantization bits $b$.
For a scalar minimum mean squared error quantizer (MMSE), the quantized signal becomes
\begin{align} 
	\label{eq:yq}
		{\bf y}_{\rm q}\!  =\! \mathcal{Q}\bigl({\rm Re}\{{\bf y}\}\bigr)  + j\mathcal{Q}\bigl({\rm Im}\{{\bf y}\}\bigr) 
	= \alpha \sqrt{p_t} {\bf h}_i{ s}+ \alpha {\bf n} +{\bf q}
\end{align} 
where $\mathcal{Q}(\cdot)$, $\alpha$, and $\bq$ are the element-wise quantizer, quantization gain, and additive quantization noise vector, respectively.
The quantization gain $\alpha$ is a function of quantization bits $b$ and is defined as $\alpha = 1- \beta $, where $\beta = {\mathbb{E}[|{y}- {y}_{{\rm q}}|^2]}/{\mathbb{E}[|{y}|^2]}$. 
{\color{black} For a scalar MMSE quantizer with Gaussian signaling ${s} \sim \mathcal{CN}({0}, 1)$, the values of $\beta$ for $b \leq 5$ \cite{max1960quantizing} are shown in Table 1 in \cite{fan2015uplink} and $\beta$ is approximated as $\beta \approx \frac{\pi\sqrt{3}}{2} 2^{-2b}$ for $b > 5$ \cite{gersho2012vector}.
The quantization noise ${\bf q}$ is uncorrelated with the quantization input ${\bf y}$ and follows ${\bf q} \sim \mathcal{CN}({\bf 0},{\bf R}_{\bf qq})$ \cite{fletcher2007robust, fan2015uplink} where }
\begin{align}
	\label{eq:cov}
	\mathbf{R}_{\bf qq}= \alpha(1-\alpha)\,{\rm diag}(p_t{\bh}_i{\bh}_i^H + \sigma^2{\mathbf{I}_{N_r}}).
\end{align}

We consider the MRC combiner $\bw = \bh_{i}$ at the receiver.
Then, the received signal becomes
\begin{align}
	\label{eq:z}
	z = \bw^H \by_{\rm q} = \alpha \sqrt{p_t}\bw^H\bh_i s + \alpha \bw^H\bn + \bw^H\bq. 
\end{align}
The achievable rate of the considered system is computed as 
\begin{align}
	\label{eq:rate}
	\cR(\alpha) = \log_2\left(1+ \frac{\alpha^2 p_t\|\bh_i\|^4}{\alpha^2\sigma^2\|\bh_i\|^2 + \bh_i^H \bR_{\bq \bq}\bh_i}\right).
\end{align}
Note that the received SNR in \eqref{eq:rate} includes the quantization noise term $\bh_i^H\bR_{\bq \bq}\bh_i$, which lowers the received SNR and makes the analysis more challenging compared to infinite-resolution ADC systems.


\section{Mathematical Preliminaries}
\label{sec:weibull}
\thispagestyle{empty}

In this section, we introduce an approximated PDF of the sum of Weibull random variables for the ergodic rate analysis with TAS.
Let $X = \sum_{i=1}^{N_r}W_i$ be the sum of $N_r$ IID Weibull random variables $W_i$ with the PDF for $w \geq 0$: 
\begin{align}
	\nonumber
	f_{W_i}(w) = \frac{k_iw^{k_i-1}}{\Omega_i}e^{-\frac{w^{k_i}}{\Omega_i}}.
\end{align}
Here, $k_i\! >\!0 $ is a shape parameter and $\Omega_i\! =\! \bbE\big[W_i^{k_i}\big]$ is a scale parameter. 
The PDF and CDF of $X$ are approximated\footnote{{\color{black}The approximated distribution of the Weibull sum is a generalized Gamma distribution $f(x;a,d,p)$ with parameters $a =\left({\Omega}/{\mu}\right)^{\frac{1}{k}}$, $d = k\mu$, and $p = k$.}}
 as \cite{yacoub2006simple}
 \begin{gather}
	\label{eq:fw}
	f_X(x) \!=\! \frac{k\mu^\mu x^{k\mu - 1}}{\Omega^\mu \Gamma(\mu)}e^{-\frac{\mu x^k}{\Omega}}, ~ 
	F_X(x) \!=\! 1\! -\! \frac{\Gamma\left(\mu,{\mu x^k}/{\Omega}\right)}{\Gamma(\mu)}
\end{gather}
where $k >0 $ is a shape parameter, $ \Omega = \bbE\big[X^k\big]$ is a scale parameter, and $\mu\! =\! \bbE^2\big[X^k \big]/\bbV\big[X^k\big]$.
Here, $\Gamma(\mu)$ is the gamma function, and $\Gamma(a,z) = \int_z^\infty t^{a-1}e^{-t}dt$ is the upper incomplete gamma function.
We use $\gamma(a,z)= \int_0^z t^{a-1}e^{-t}dt$ to denote the lower incomplete gamma function. 
{\color{black} Note that $\Gamma(a,z)$ for $a=0$ is equivalent to the exponential integral function $E_1(z)$.}
{\color{black}The parameters, $k, \mu$ can be computed by solving \cite{yacoub2006simple}
\begin{gather}
	\nonumber
	\Gamma^2(\mu + {1}/{k})\bbE[X^2] = \Gamma(\mu)\Gamma(\mu + {2}/{k})\bbE^2[X]	\\
	\nonumber
	\Gamma^2(\mu + {2}/{k})\bbE[X^4] = \Gamma(\mu)\Gamma(\mu + {4}/{k})\bbE^2[X^2]
\end{gather}
where the moment $\bbE[X^n]$ is given by
\begin{align}
	\nonumber
	\bbE[X^n] = &\sum_{n_1=0}^{n}\sum_{n_2=0}^{n_1}\cdots\sum_{n_{N_r-1}=0}^{n_{N_r-2}}\left(\begin{matrix}n \\ n_1\end{matrix}\right)\left(\begin{matrix}n_1 \\ n_2\end{matrix}\right)\cdots\left(\begin{matrix}n_{N_r-2} \\ n_{N_r -1}\end{matrix}\right)\\
		\nonumber
		& \times \bbE[W_1^{n-n_1}]\bbE[W_2^{n_1-n_2}] \cdots \bbE[W_{N_r}^{n_{N_r}-1}].
\end{align}
Here, $n$ is a positive integer and $\bbE[W_i^n] = \Omega_i^{{n}/{k_i}}\Gamma\left(1+{n}/{k_i}\right) $.
Then, $\Omega$ is obtained as $\Omega^{1/k} = {\mu^{1/k}\Gamma(\mu)E[X]}/{\Gamma(\mu + 1/k)}$.}

\section{Transmit Antenna Selection}
\label{sec:DL_txantenna}

{\color{black} In this section, we analyze the ergodic rate of the TAS system assuming that the channel state information is known perfectly at both the transmitter and receiver.}
Based on the achievable rate in \eqref{eq:rate}, the ergodic rate of TAS with the MRC receiver under quantization error is	
\begin{align}
	\label{eq:Rtas}
	\bar \cR_{\rm t}(\alpha) 
	\!=\! \bbE\!\!\left[\max_{i\in \mathcal{T}} \log_2\!\left(\!1\!+\! \frac{\rho\alpha \|\bh_i\|^4}{\|\bh_i\|^2\! +\! \rho(1\!-\!\alpha)\!\sum_{j=1}^{N_r}\! |h_{j,i}|^4}\!\right)\!\right]\!\!\!
\end{align}
where $\rho = p_t/\sigma^2$ is the transmit SNR, $h_{j,i}$ is the $j$th element of $\bh_i$, and $\cT = \{1,\dots,N_t\}$.
{\color{black} In this paper, we specify the number of transmit antennas and receive antennas $N_t,N_r$ as $\bar{\cR}_t(\alpha;N_t,N_r)$ if necessary.}
{\color{black}
\begin{lemma}
	\label{lem:lowerbound}
	The ergodic rate of TAS with MRC in \eqref{eq:Rtas} is lower bounded by
	\begin{align}
		\label{eq:lowerbound}
		\bar \cR_{\rm t}(\alpha)\! \geq\! \bbE\!\left[ \log_2\!\left(\!1\!+\! \frac{\rho\alpha \|\bh_{i^\star}\|^4}{\|\bh_{i^\star}\|^2 + \rho(1\!-\!\alpha)\!\sum_{j=1}^{N_r}\! |h_{j,{i^\dagger}}\!|^4}\!\right)\!\right]
	\end{align}
where $i^\star\!=\! \argmax_{i\in \mathcal{T}}\!\|\bh_i\|$ and ${i^\dagger} \!=\! \argmax_{i \in \mathcal{T}} \!\sum_{j=1}^{N_r} |h_{j,i}|^4$.
\end{lemma}
\begin{proof}
	 The ergodic rate in \eqref{eq:Rtas} is lower bounded by the ergodic rate with selecting the antenna $i^\star\!=\! \argmax_{i\in \mathcal{T}}\!\|\bh_i\|$ since selecting the antenna $i^\star$ may not be optimal due to the quantization error term $\rho(1\!-\!\alpha)\!\sum_{j=1}^{N_r} \!|h_{j,{i^\star}}|^4$. 
	The ergodic rate with the antenna $i^\star$ is further lower bounded by considering ${i^\dagger}\! =\! \argmax_{i\in\mathcal{T}} \!\sum_{j=1}^{N_r}\! |h_{j,i}|^4$ for the quantization error term because of the inequality: $\bbE\big[\sum_{j}\! |h_{j,i^\star}|^4\big] \leq \bbE\big[\sum_{j}\! |h_{j,{i^\dagger}}|^4\big]$.
\end{proof}
Note that the lower bound in \eqref{eq:lowerbound} becomes an exact expression when there is not any quantization error, i.e., the quantization gain is $\alpha = 1$.
Accordingly, the lower bound becomes tighter as the number of quantization bits $b$ increases. 
Using  \eqref{eq:lowerbound}, we derive the approximated lower bound of the ergodic rate with TAS in closed form.
In addition, we show that the derived rate provides a good lower bound for a small number of quantization bits in Section~\ref{sec:simulation}.  }
\begin{theorem}
	\label{thm:ER1}
	The lower bound of the ergodic rate of TAS with MRC under coarse quantization is approximated as
	\begin{align}
		\label{eq:R_LB}
		\bar{\cR}^{\rm lb}_{\rm t}(\alpha) \approx \log_2\!\left(\!1+ \frac{\rho\alpha \int_0^\infty zG(z)dz}{\int_0^\infty (G(z)+ \rho(1-\alpha)G_q(z))dz}\!\right)
	\end{align}
	where 
	\begin{gather}
		\nonumber
		G(z)=	e^{-z}z^{N_r}\left(\frac{\gamma(N_r, z)}{\Gamma(N_r)}\right)^{N_t -1}\\
		\nonumber
		G_q(z) = \frac{2\Gamma(N_r+1)}{\Gamma\left(\mu+\frac{1}{k}\right)} e^{-z}z^{\mu + \frac{1}{k} -1}\left(\frac{\gamma(\mu, z)}{\Gamma(\mu)}\right)^{N_t -1}. 
	\end{gather}
\end{theorem}
\begin{proof}
	From Lemma~\ref{lem:lowerbound}, the ergodic rate in \eqref{eq:Rtas} is lower bounded and further approximated as
	\begin{align}
		\nonumber
		&\bar \cR_{\rm t}(\alpha) \geq \bbE\!\left[ \log_2\!\left(\!1+ \frac{\rho\alpha \|\bh_{i^\star}\|^4}{\|\bh_{i^\star}\|^2 + \rho(1-\alpha)\sum_{j=1}^{N_r} |h_{j,{i^\dagger}}|^4}\!\right)\!\right]\\
		\label{eq:LB_proof1}
		&\!\!\stackrel{(a)}\approx \log_2\!\left(\!1+ \frac{\rho\alpha \bbE\left[\|\bh_{i^\star}\|^4\right]}{ \bbE\!\left[\|\bh_{i^\star}\|^2\right]\! +\!  \rho(1-\alpha)\bbE\!\left[\sum_{j=1}^{N_r} |h_{j,{i^\dagger}}|^4\right]}\!\right)
	\end{align}
	where $i^\star\!=\! \argmax_i\!\|\bh_i\|$, ${i^\dagger}\! =\! \argmax_i\! \sum_j \!|h_{j,i}|^4$, and  $(a)$ follows from Lemma 1 in \cite{zhang2014power}.
	Since $\|\bh_i\|^2 \!\sim\! \chi^2_{2N_r}$ follows the IID chi-squared distribution with $2N_r$ degrees of freedom, the PDF of $Y \!=\! \|\bh_{i^\star}\|^2$ is derived by using order statistics as
	\begin{align}
		\label{eq:fy}
		f_Y(y) = \frac{N_t}{(N_r - 1)!} \left(1-e^{-y}\sum_{n=0}^{N_r-1}\frac{y^n}{n!}\right)^{\!\!\!N_t -1}\!\!\!\!\!\! y^{N_r -1}e^{-y}.
	\end{align}  
	Using \eqref{eq:fy} with $\Gamma(N_r,y)/\Gamma(N_r) = e^{-y}\sum_{n=0}^{N_r-1}{y^n}/{n!}$, we compute the expectations in \eqref{eq:LB_proof1} as
	\begin{gather}
		\label{eq:Ey}
		\bbE\!\Big[\|\bh_{i^\star}\|^2\Big]\!\! =\!\! \frac{N_t}{(N_r-1)!}\!\!\int_0^\infty \!\!\!\left(\!1\!-\!\frac{\Gamma(N_r,y)}{\Gamma(N_r)}\!\right)^{\!\!N_t-1}\!\!\!\!\!\!\!\! y^{N_r} e^{-y}dy\\
		\label{eq:Ey2}
		\bbE\!\Big[\|\bh_{i^\star}\|^4\Big]\!\! =\!\! \frac{N_t}{(N_r\!-\!1)!}\!\!\int_0^\infty\! \!\! \left(\!1\!-\!\frac{\Gamma(N_r,y)}{\Gamma(N_r)}\!\right)^{\!\!N_t-1}\!\!\!\!\!\!\!\!y^{N_r+1} e^{-y}dy.\!\!
	\end{gather}
	To compute $\bbE\!\big[\sum_{j=1}^{N_r}\! |h_{j,{i^\dagger}}|^4\big]$ in \eqref{eq:LB_proof1}, the PDF of $\sum_{j=1}^{N_r}\! |h_{j,{i^\dagger}}|^4$ is necessary.
	Since $|h_{j,i}|^4$ is equal to a square of an exponential random variable with a rate parameter $\lambda \!=\! 1$, $|h_{j,i}|^4$ follows the Weibull distribution with a scale parameter $\Omega_{j,i}\! =\! 1$ and shape parameter $k_{j,i} \!=\! 1/2$, i.e., $W_{j,i} \!\sim \! {\rm Wei}(1,{1}/{2})$.
	Let $X_i  \!=\! \sum_{j=1}^{N_r} \!|h_{j,i}|^4$. 
	Then, the approximated PDF and CDF of $X_i$ are given in \eqref{eq:fw},
	and the PDF of $Z \!=\! X_{i^\dagger} \!=\! \max_i X_i$ is derived as
	\begin{align}
		\label{eq:fz}
		f_Z(z)\! =\! \frac{N_t k \mu^\mu}{\Omega^\mu \Gamma(\mu)} \!\!\Bigg(\!1- \frac{\Gamma\left(\mu,\frac{\mu z^k}{\Omega}\right)}{\Gamma(\mu)}\!\Bigg)^{\!\!N_t - 1} \!\!\!\!\!\!\!\!z^{k\mu-1} e^{-\frac{\mu z^k}{\Omega}}.
	\end{align}
	Replacing $y = \mu z^k/\Omega$, we derive $\bbE\big[\sum_{j=1}^{N_r} |h_{j,{i^\dagger}}|^4\big]$ as
	\begin{align}	
		\label{eq:Ez}
		\bbE[Z] \!\approx \! \frac{N_t }{\Gamma(\mu)}\!\left(\frac{\Omega}{\mu}\right)^{\!\!\frac{1}{k}}\!\! \int_0^\infty \!\!\left(\!1- \frac{\Gamma(\mu,y)}{\Gamma(\mu)}\!\right)^{\!\!N_t - 1} \!\!\!\!\!\!\!\!y^{\mu + \frac{1}{k} -1} e^{-y} dy.
	\end{align}
	Putting \eqref{eq:Ey}, \eqref{eq:Ey2}, and \eqref{eq:Ez} into \eqref{eq:LB_proof1}, the ergodic rate in \eqref{eq:LB_proof1} becomes \eqref{eq:R_LB} after simplification by using the definition of $\Omega^{1/k} = \mu^{1/k}\Gamma(\mu)E[X]/\Gamma(\mu + 1/k)$ \cite{yacoub2006simple}.
	Note that the expectation of the sum of $N_r$ Weibull random variables is $E[X] = 2N_r$ for the Weibull distribution with $\Omega_i = 1$ and $k_i = 1/2$. 
	This completes the proof.
\end{proof}
\thispagestyle{empty}
Theorem \ref{thm:ER1} derives the approximated lower bound of ergodic rate with TAS for multiple receive antennas as a function of system parameters.
{\color{black} Note that \eqref{eq:R_LB} becomes more accurate as the number of receive antennas $N_r$ increases \cite{zhang2014power}.}
We further derive the approximated ergodic rate of TAS with a single receive antenna $\bar \cR_t(\alpha;N_t,1)$ in closed form.
\begin{corollary}
	The ergodic rate of TAS for a single receive antenna under coarse quantization is approximated as
	\begin{align}
		\label{eq:R1}
		\bar{\cR}_{\rm t}(\alpha;N_t,1) \approx \log_2\left(1+ \frac{\rho\alpha \sum_{n=1}^{N_t}\frac{1}{n}}{1+\rho(1-\alpha)\sum_{n=1}^{N_t}\frac{1}{n}}\right).
	\end{align}
\end{corollary}
\begin{proof}
For $N_r =1$, the ergodic rate of TAS in \eqref{eq:Rtas} is given as
\begin{align}
	\nonumber
	\bar \cR_{\rm t}(\alpha;N_t,1) &= \bbE\!\left[\max_{i\in \mathcal{T}} \log_2\!\left(\!1+ \frac{\rho\alpha |h_i|^2}{1+ \rho(1-\alpha) |h_{i}|^2}\!\right)\!\right]\\
	\label{eq:R1_proof1}
	& \stackrel{(a)}= \bbE\!\left[ \log_2\!\left(\!1+ \frac{\rho\alpha |h_{i^\star}|^2}{1+ \rho(1-\alpha) |h_{i^\star}|^2}\!\right)\!\right]
\end{align}
where $(a)$ comes from $i^\star = \max_i |h_i|$. 
Similarly to the proof of Theorem \ref{thm:ER1}, the ergodic rate in \eqref{eq:R1_proof1} can be approximated by using Lemma 1 in \cite{zhang2014power}.
Computing $\bbE[Y]$ with $Y = |h_{i^\star}|^2$ and $f_Y(y; N_r = 1)$ in \eqref{eq:fy}, we derive the ergodic rate in \eqref{eq:R1}.
\end{proof}
\begin{corollary}
	The ergodic rate in \eqref{eq:R1} is further approximated in the large transmit antenna regime as
	\begin{align}
		\label{eq:R1_approx}
		\bar{\cR}_{\rm t}(\alpha;N_t,1)\! \approx \! \log_2\left(1+ \frac{\rho\alpha\big(\ln N_t + \gamma_e\big)}{1+\rho(1-\alpha)\big(\ln N_t + \gamma_e\big)}\right)
	\end{align}
		 where $\gamma_e$ is the Euler\--Mascheroni constant.
\end{corollary}
\begin{proof}
	The harmonic series is equivalent to $\sum_{n=1}^{N_t}{1}/{n} = \ln N_t + \gamma_e + \epsilon_{N_t}$ where $\epsilon_{N_t} \sim \frac{1}{2N_t}$.
	Thus, $\sum_{n=1}^{N_t}{1}/{n} \approx \ln N_t + \gamma_e$ for large $N_t$, which leads to \eqref{eq:R1_approx}.
\end{proof}
From \eqref{eq:R1} or \eqref{eq:R1_approx}, as the number of transmit antennas $N_t$ increases, the ergodic rate $\bar{\cR}_t(\alpha;N_t,1)$ increases, which can be considered as the TAS gain. 
{\color{black} To compare with the ergodic rate of TAS, Theorem \ref{thm:ER_siso} derives an upper bound of the ergodic rate with coarse quantization for a single transmit and single receive antenna $\bar{\cR}(\alpha;1,1)$, i.e., the special case of TAS in which there is only one transmit antenna to be selected. }
\begin{figure*}[!t]
\centering
$\begin{array}{c c}
{\resizebox{0.9\columnwidth}{!}
{\includegraphics{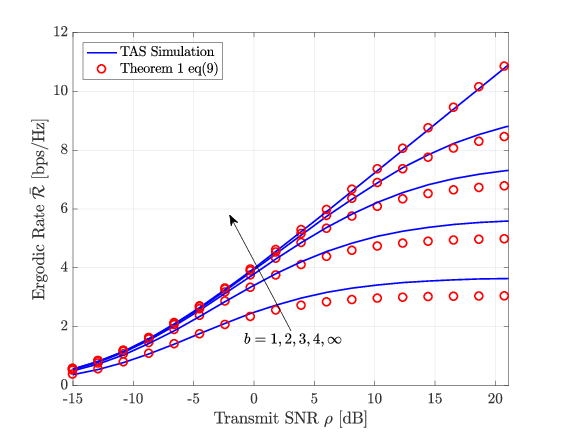}}
}&
{\resizebox{0.9\columnwidth}{!}
{\includegraphics{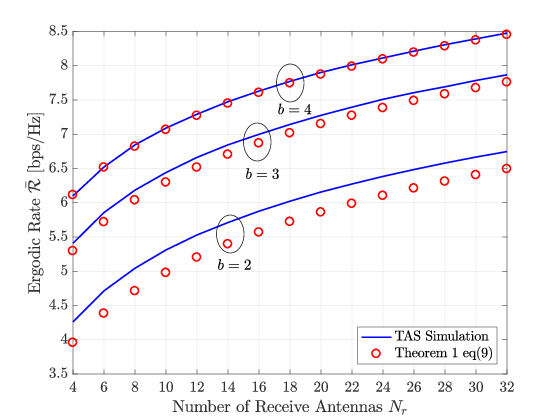}}
}\\ \mbox{\small (a)}&\mbox{\small (b)}
\end{array}$
\caption{
The ergodic rate of single-TAS for $N_t = 32$ transmit antennas with (a) $N_r = 8$ receive antennas and $b\! \in\! \{1,2,3,4\}$ quantization bits, and with (b)  $\rho = 10$ dB transmit SNR and $b\! \in\! \{2,3,4\}$ quantization bits.} 
\label{fig:theorem}
\end{figure*}

\begin{theorem}
	\label{thm:ER_siso}
	The ergodic rate with a single transmit/receive antenna under coarse quantization is upper bounded by
	\begin{align}
		\label{eq:Rrand_UB}
		\bar{\cR}(\alpha;1,1) \leq  \log_2\left(1+ \frac{\rho\alpha}{1+\rho(1-\alpha)}\right).
	\end{align}
\end{theorem}
\begin{proof}
	The ergodic rate for $N_t\! =\! N_r\! =\!1$ with quantization error is given by
	\begin{align}
		\nonumber
		\bar{\cR}(\alpha;1,1) &= \bbE\left[\log_2\left(1+\frac{\rho\alpha|h|^2}{1+\rho(1-\alpha)|h|^2}\right)\right]\\
		\nonumber
		& \stackrel{(a)} = \frac{1}{\ln 2}\left(e^{\frac{1}{\rho}}\Gamma\left(0,\frac{1}{\rho}\right) -e^{\frac{1}{\rho\beta}}\Gamma\left(0,\frac{1}{\rho\beta}\right) \right)
	\end{align}
	where $(a)$ comes from the fact that $|h|^2$ follows the exponential distribution with unit mean, i.e., $|h|^2 \sim {\rm Exp}(1)$, and $\beta = (1-\alpha)$.
	Now, to prove \eqref{eq:Rrand_UB}, we need to show \eqref{eq:inequality} for $\rho >0$.
	\begin{align}
		\label{eq:inequality}
		\ln(1\!+\!\rho)\! -\! \ln(1\!+\!\rho\beta) \geq e^{\frac{1}{\rho}}\Gamma\left(0,\frac{1}{\rho}\right) \!-\! e^{\frac{1}{\rho\beta}}\Gamma\left(0,\frac{1}{\rho\beta}\right).
	\end{align}
	To this end, let $g_1(x) = \ln(1+x) - e^{1/x}\Gamma(0,1/x)$.
	Then, \eqref{eq:inequality} reduces to $g_1(\rho) \geq g_1(\rho \beta)$.
	Since $0<\beta<1$, if $g_1(x)$ is monotonically increasing for $x>0$, \eqref{eq:inequality} holds true, or equivalently, we need to show that $g_2(y) = \ln(1+1/y) - e^{y}\Gamma(0,y)$ is monotonically decreasing for $y>0$.
	The derivative of $g_2(y)$ is given as $g_2'(y) = {1/({1+y})-e^{y}\Gamma(0,y)}$.
	Accordingly, we need to show $g_2'(y) < 0  \Leftrightarrow \frac{e^{-y}}{1+y} < \Gamma(0,y)$ for $y>0$.
	
	To show $\frac{e^{-y}}{1+y} < \Gamma(0,y)$, we use the following inequality: $ \frac{1}{2} e^{-y}\ln(1+2/y) < \Gamma(0,y) $ \cite{abramowitz1964handbook}.
	Then, we need to show the following inequality \eqref{eq:inequality2} to prove $g_2'(y)<0$:
	\begin{align}
		\label{eq:inequality2}
		\frac{e^{-y}}{1+y} < \frac{1}{2} e^{-y}\ln\left(1+\frac{2}{y} \right).	
	\end{align}
	Let $h(y) = \frac{1}{2} (1+y) \ln(1+{2}/{y}) -1$.
	Then, \eqref{eq:inequality2} is equivalent to $h(y) > 0$. 
	Since  $\lim_{y \to 0} h'(y)= -\infty$, $\lim_{y \to \infty} h'(y) = 0$, and $h''(y) > 0$, we have $h'(y) <0$ for $y>0$, i.e., $h(y)$ is monotonically decreasing for $y>0$.
	As $\lim_{y \to 0} h(y)\!=\! \infty$ and $\lim_{y \to \infty} h(y)\! =\! 0$, this proves $h(y)\!>\! 0$, which also proves \eqref{eq:inequality2}.
	Therefore, \eqref{eq:inequality} holds true and this completes the proof.
\end{proof}
\thispagestyle{empty}

We note that due to the TAS gain, the ergodic rate of TAS $\bar \cR_t(\alpha;N_t,1)$ can achieve a higher rate than the upper bound of the ergodic rate with a single transmit antenna in \eqref{eq:Rrand_UB} as the number of transmit antennas $N_t$ increases.
The increase of ergodic rate from the TAS gain for a single receive antenna is limited to $\log_2(1\!+\!\alpha/(1\!-\!\alpha))$ as $N_t\! \to\! \infty$.
This is because selecting the antenna with the larger channel gain also increases the quantization noise variance.
Although the TAS gain under coarse quantization is limited when compared to the TAS gain for perfect quantization, the TAS gain can still provide a large increase of ergodic rate since $\bar \cR_t(\alpha;N_t,1)$ with $N_t \!\to\! \infty$ converges to $\log_2(1\!+\!\alpha/(1\!-\!\alpha))$ which the upper bound of $\bar \cR(\alpha;1,1)$ requires $\rho \to \infty$ to achieve.
\section{Simulation Results}
\label{sec:simulation}

%


In this section, we validate the derived ergodic rates and resulting intuition.
In Fig. \ref{fig:theorem}(a), the approximated lower bound of the ergodic rate of TAS in \eqref{eq:R_LB} is compared with simulation results with respect to the transmit SNR $\rho$ for $N_t = 32$, $N_r = 8$, and $b \in \{1,2,3,4,\infty\}$.
We note that the lower bound in \eqref{eq:R_LB} shows a small gap from the simulation results.
In particular, the gap decreases as $b$ increases or $\rho$ decreases since the quantization error becomes less dominant than the AWGN.  
Accordingly, \eqref{eq:R_LB} can serve as the approximation of the ergodic rate of TAS in the less dominant quantization error regime.
As $b$ increases, the ergodic rate convergence occurs in the higher transmit SNR regime.
For $b =\infty$, the ergodic rate keeps increasing without convergence as $\rho$ increases since the rate is not limited by quantization error.

In Fig. \ref{fig:theorem}(b), we further compare the lower bound in \eqref{eq:R_LB} with the simulation results with respect to $N_r$ for $N_t = 32$, $\rho = 10$ dB, and $b \in \{2,3,4\}$.
For a different number of receive antennas $N_r$, the derived lower bound in \eqref{eq:R_LB} shows a small gap. 
As discussed, the gap decreases as $b$ increases, which makes the quantization error less dominant.
Thus, we again conjecture that \eqref{eq:R_LB} can serve as the approximation of the ergodic rate of TAS in the high-resolution ADC regime.

In Fig. \ref{fig:corollary}(a), the approximated ergodic rates of TAS for $N_r =1$ in \eqref{eq:R1} and \eqref{eq:R1_approx} are compared with simulation results with respect to $\rho$ for $N_t\in\{4, 16, 64\}$ and $b=3$.
We further compare the derived upper bound of the ergodic rate  for $N_t = N_r=1$ in \eqref{eq:Rrand_UB} with simulation results for $b=3$.
The ergodic rates in \eqref{eq:R1} and  \eqref{eq:R1_approx} accurately align with the simulation results, and \eqref{eq:Rrand_UB} provides a valid upper bound for $\bar \cR(\alpha;1,1)$.
We note that the system with TAS achieves a large improvement in ergodic rate from the system with a single transmit antenna ($N_t = 1$).
In addition, as the number of transmit antenna increases, the TAS case achieves higher ergodic rate owing to the TAS gain as shown in \eqref{eq:R1} and \eqref{eq:R1_approx}.

To further analyze the TAS gain, we evaluate the ergodic rate with respect to $N_t$ in Fig. \ref{fig:corollary}(b).
For a small or medium number of antennas, the ergodic rate rapidly increases as $N_t$ increases. 
In the large number of antenna regime, however, the rate of increase becomes slower.
Since the received SNR in low-resolution ADC systems is lower than that in infinite-resolution ADC systems due to quantization error, the increase of the received SNR owing to the TAS gain can provide large improvement of ergodic rate in the small or medium number of antenna regime.
{\color{black} For example, the ergodic rate with TAS for $N_t=12$ transmit antennas in Fig.~\ref{fig:corollary}(b) shows $1.87\times$ increase from the single transmit antenna $N_t =1$ case $\bar\cR(\alpha;1,1)$ in Fig.~\ref{fig:corollary}(a) for $b = 3$ and $\rho = 5$ dB.}
Consequently, we can use a moderate number of transmit antennas to obtain a proper TAS gain in low-resolution ADC systems.
Overall, the simulation results validate the derived ergodic rates and confirm the TAS gain for low-resolution ADC systems.

\begin{figure}[!t]
\centering
$\begin{array}{c c}
{\resizebox{0.9\columnwidth}{!}
{\includegraphics{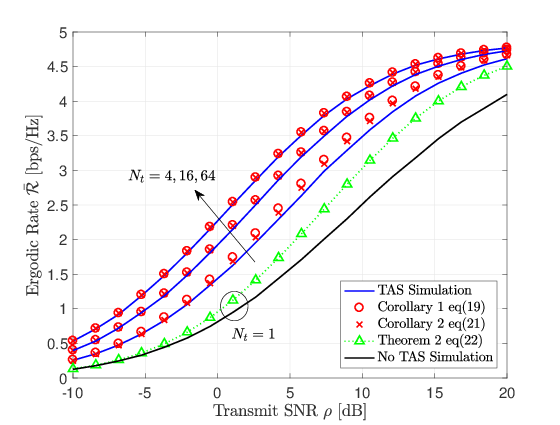}}
}\\\mbox{\small (a)} \\
{\resizebox{0.9\columnwidth}{!}
{\includegraphics{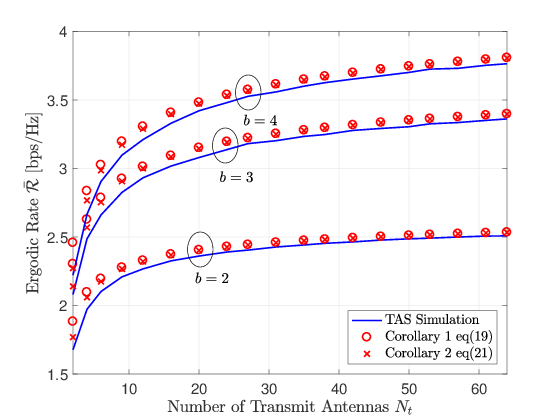}}
}\\\mbox{\small (b)}
\end{array}$
\caption{
The ergodic rate of TAS for $N_r =1$ receive antenna with (a) $N_t \in \{4, 16, 64\}$ transmit antennas and $b = 3$ quantization bits, and with (b) $\rho= 5$ dB transmit SNR and $b \in \{2,3,4\}$ quantization bits.} 
\label{fig:corollary}
\end{figure}

\section{Conclusion}
\thispagestyle{empty}

We investigate single transmit antenna selection in low-resolution ADC systems.
We adopt the approximated PDF and CDF of the sum of Weibull distributions to address the challenge in analyzing the ergodic rate with quantization error.
Leveraging the approximated distribution and order statistics,
we derive the approximated lower bound of the ergodic rate with TAS.
For a single receive antenna case, we further derive the approximated ergodic rate with TAS in closed form.
The analysis shows that
the TAS method achieves large improvement with a moderate number of transmit antennas in ergodic rate.
The simulation results validate the derived ergodic rate expressions and intuition regarding the TAS method in low-resolution ADC systems.
{\color{black} The ergodic rate analysis of the TAS method under the presence of correlated channels or imperfect channel state information is desirable for future work.}

\bibliographystyle{IEEEtran}
\bibliography{P2PAntSel.bib}
\end{document}